\documentclass[10pt,twoside]{amsart}

\usepackage{amscd}
\usepackage{amssymb} 
\usepackage{verbatim}
\usepackage{amsmath}
\usepackage{amsxtra}
\usepackage[noadjust]{cite}         

\newtheorem{theorem}{Theorem}[section]

\newtheorem{proposition}[theorem]{Proposition}

\theoremstyle{definition}

\numberwithin{equation}{section}

\begin{document}

\title[Stationary solutions of NLSE in supercritical dimensions]{%
  Stationary solutions of semilinear Schr\"{o}dinger
  \\[1mm] equations with trapping potentials
  \\[1mm] in supercritical dimensions}

% AUTHOR1
\author[F. Ficek]{Filip Ficek}
\address{Institute of Theoretical Physics, Jagiellonian University,
  \L{}ojasiewicza 11, 30-348 Krak\'{o}w, Poland} 
  \email{ficekf@gmail.com}  

\thanks{This research was funded by the Polish National Science Centre within Grants No. 2020/36/T/ST2/00323 and No. 2017/26/A/ST2/00530.}

\begin{abstract}
  Nonlinear Schr\"{o}dinger equations are usually investigated with the use of the variational methods that are limited to energy-subcritical dimensions. Here we present the approach based on the shooting method that can give the proof of existence of the ground states in critical and supercritical cases. We formulate the assumptions on the system that are sufficient for this method to work. As examples, we consider Schr\"{o}dinger-Newton and Gross-Pitaevskii equations with harmonic potentials.
\end{abstract}



\maketitle

% DOCUMENT

\section{Introduction} \label{S:intro}
The most common approach in the study of nonlinear Schr\"{o}dinger equations (NLSE) is based on the variational methods. However, since these methods rely on some compactness results, they cease to work in energy-supercritical dimensions. From the application point of view, it does not seem to pose a great problem because such equations are usually used to describe various quantum-mechanical systems that are at most three-dimensional. An example of such NLSE is the Schr\"{o}dinger-Newton-Hooke  equation (SNH) that describes a self-gravitating quantum gas in a harmonic trap:
\begin{align}
\begin{cases}\label{eqn:SNHt}
    i\,\partial_t\psi&=-\Delta \psi+|x|^2\psi+v\psi,\\
    \Delta v&=|\psi|^2,
\end{cases}
\end{align}
where $\psi$ is the wavefunction and the nonlinearity is introduced by the gravitational potential $v$. In \cite{Biz18} the authors showed that this system can be also obtained as a nonrelativistic limit of the perturbations of the anti-de Sitter spacetime. This result connects it to one of the most important open problems in mathematical general relativity, the stability of anti-de Sitter spacetime \cite{Biz13}, and gives a motivation to investigate Eq.\ (\ref{eqn:SNHt}) in higher dimensions.

The literature regarding NLSE with trapping potentials (potentials that diverge to $\infty$ as $|x|\to\infty$) in supercritical dimensions is rather scarce. Existence of a whole family of stationary solutions (solutions satisfying the ansatz $\psi(t,x)=e^{-i\omega t}u(x)$, where $\omega$ is some real value called the frequency and $u$ is a real function vanishing at infinity) of SNH was shown in \cite{Fic21}. The only other similar system that was investigated in supercritical dimensions was the Gross-Pitaevskii equation with a harmonic potential (GP) \cite{Biz21, Pel, Sel11, Sel12, Sel13}:
\begin{align}\label{eqn:gpt}
    i\, \partial_t \psi=-\Delta \psi + |x|^2 \psi-|\psi|^2 \psi.
\end{align}

The goal of this short paper is to formulate a common framework that can be used for various semilinear Schr\"{o}dinger equations with trapping potentials in supercritical dimensions. In Section \ref{S:main} we describe our approach, which is based on a reduction to the ordinary differential equation and application of the shooting method. We state the necessary assumptions and prove the existence of ground states for systems satisfying them. Section \ref{S:app} shows how this result can be applied to SNH and GP equations. Finally, in Section \ref{S:conclusions} we outline the possible extensions and future prospects.

\section{Main result} \label{S:main}
Since we are interested in stationary solutions, the initial Schr\"{o}dinger equation is reduced to a nonlinear elliptic equation with a trapping potential. Ground states of such equations are usually spherically symmetric \cite{Li93}, letting us to write down the problem as an ODE, typically having the form
\begin{align}\label{eqn:NLS}
    -u''-\frac{d-1}{r} u'+V(r)u-F(r,u(r))=\omega\, u,
\end{align}
where $r=|x|$, $u(r)$ is the solution we seek, $V$ is the trapping potential (i.e.\ $\lim_{|x|\to\infty}V(x)=\infty$), and $F$ denotes the nonlinearity. When looking for stationary solutions, one usually specifies some characteristics of the sought solution, e.g., its frequency or mass. Here we will be looking for ground states $u$ with some fixed value in the center of the symmetry $u(0)=b>0$. Regularity of the solution implies $u'(0)=0$. Since $u$ is a ground state, we also require that $\lim_{r\to\infty}u(r)=0$ and $u(r)>0$. Our final goal is to prove that there is such $\omega$ that there exists a solution $u$ satisfying the conditions above and Eq.\ (\ref{eqn:NLS}) -- the ground state with frequency $\omega$.

In principle, one could now try to employ the shooting method with $\omega$ as the shooting parameter. However, sometimes one has a better control on some other quantity related to $\omega$ that we will denote by $c$. Let us then rewrite the above equation as the following initial value problem
\begin{align}
\begin{cases}\label{eqn:init}
    u''+\frac{d-1}{r} u'-V(r)u+F_c(r,u(r))=0,\\
    u(0)=b>0,\quad u'(0)=0,
\end{cases}
\end{align}
where $F_c$ contains the nonlinearity and depends continuously on some parameter $c$. One can easily show that the singularity at $r=0$ present in this equation does not pose any problem and all classical results regarding existence, uniqueness, and continuous dependence of the solutions still hold \cite{Has}. Hence, for any fixed value of $c$, we get some function $u_c(r)$ with its maximal domain $[0,R_c)$, where $R_c$ may be infinite. 

The proof of the existence of the ground states for SNH presented in \cite{Fic21} is following this line of action and then relies on the analysis of the behaviour of $u_c$ as $c$ changes. For the sake of generality, it may be convenient to perform here such analysis in isolation from the initial ODE-based context. We can just see the set of solutions of Eq.\ (\ref{eqn:init}) for various $c$ as the family of functions $\{u_c\}$ depending continuously on a single parameter $c$ (their derivatives $u_c'$ also depend continuously on $c$ since they are solutions to the second order ODE). Assume then that this family satisfies the following six conditions.
\begin{enumerate}
    \item[(A1)] There is a value of $c$ such that the function $u_c$ has $r_0$ at which $u_c(r_0)=0$ while $u_c(r)>0$ and $u_c'(r)<0$ for $r\in(0, r_0)$.
    \item[(A2)] The function $u_0$ is positive.
    \item[(A3)] For any $c$, if at some point $r_0$ it holds $u_c(r_0)=u_c'(r_0)=0$, then $u$ is identically zero.
    \item[(A4)] Functions $u_c$ cannot have an inflection point while they are positive and decreasing.
    \item[(A5)] It holds $u_c''(0)<0$ for $c>0$.
    \item[(A6)] For any $c$, it either holds $\lim_{r\to R_c}u_c(r)=\infty$, $\lim_{r\to R_c}u_c(r)=-\infty$, or $\lim_{r\to R_c}u_c(r)=0$, where $R_c$ may be infinite.
\end{enumerate}
As we prove now it leads to the existence of such $c_0$ that $u_{c_0}$ is the ground state of our problem. It means that for a generic problem such as (\ref{eqn:init}), it is enough to check whether the solutions satisfy these conditions to show that the ground state exists. This is the approach we employ in Section \ref{S:app}.

\begin{theorem} \label{T:1}
Let $\{u_c|c\geq 0\}$ be a family of at least twice differentiable functions with domains $[0,R_c)$ satisfying $u_c(0)=b>0$ and $u_c'(0)=0$. Let the values of $u_c$ and $u_c'$ depend pointwise-continuously on $c$. Then if this family satisfies (A1)-(A6), there exists $c_0$ such that $u_{c_0}$ is a positive function on domain $[0,\infty)$ and decreasing to zero at infinity.
\end{theorem}

\begin{proof}
	Let us introduce a set of values of parameter $c$ defined by the behaviour of $u_c$:
	\begin{align*}%\label{eqn:I0}
	    I=\left\{c\geq 0\,|\, \exists\, r_0>0: u_c(r_0)=0 \mbox{ while } u_c(r)>0 \mbox{ and } u_c'(r)<0 \mbox{ for } r\in(0, r_0)\right\}.
	\end{align*}
	Assumption (A1) tells us that this set is not empty, so $c_0=\inf I$ is finite, while (A2) implies that $0\not\in I$. We claim that $u_{c_0}$ is the sought function. The main tool in this proof will be the continuous dependence of $u_c$ and $u_c'$ on the parameter $c$.
	
	Assume $c_0>0$ for now. If $u_{c_0}$ crosses zero at some point, let us denote the first such occurrence by $r_0$. Then $u_{c_0}$ must do it transversally due to (A3). It means that there exists $U$ -- a neighbourhood of $c_0$ such that for all values of $c$ in it, $u_c$ is also crossing zero. Additionally, since $c_0=\inf I$, thanks to (A4) and (A5) the functions $u_c$ for every $c$ in $U$ must be decreasing up to the crossing with zero (because no new stationary point may appear between $r=0$ and the first crossing as $c$ slightly changes). It means that $U\subset I$ and as a result $c_0$ cannot be the infimum of $I$. The fact that $u_{c_0}$ cannot cross zero rules out the possibility that $\lim_{r\to\infty}u_{c_0}(r)=-\infty$.
	
	Let us now assume that $\lim_{r\to\infty}u_{c_0}(r)=\infty$. Condition (A5) tells us that $u_{c_0}$ is initially decreasing, so there must be a point where $u_{c_0}$ has the first positive minimum. This time the continuous dependence of $u'_{c_0}$ on $c$ tells us that for some small neighbourhood of $c_0$ functions $u_c$ also have such a minimum, thanks to (A4). It contradicts $c_0$ being the infimum of $I$, so it must hold $u'_{c_0}(r)<0$. Similar analysis also applies to the case of $c_0=0$. 
	
	As $u_{c_0}$ cannot diverge to any of the infinities, the trichotomy (A6) tells us that $\lim_{r\to\infty}u_{c_0}(r)=0$. We additionally know that $u_{c_0}(r)>0$ and $u'_{c_0}(r)<0$, so $u_{c_0}$ is a positive decreasing function.
\end{proof}

\section{Applications} \label{S:app}
Now we briefly show how one can apply Theorem \ref{T:1} to show the existence of the ground states in the cases of two different semilinear Schr\"{o}dinger equations: (\ref{eqn:SNHt}) and (\ref{eqn:gpt}). As already noted, we will be looking for solutions $u$ with some fixed central value $u(0)=b>0$.

\subsection{Schr\"{o}dinger-Newton-Hooke equation} \label{SS:SNH}

As the ground state of Eq. (\ref{eqn:SNHt}) we understand a stationary solution with both $u$ and $v$ tending to zero at infinity while $u$ stays positive. Such solutions must be spherically symmetric as shown in \cite{Bus00}. It leads to a system of two ODEs for which the shooting method approach may seem problematic at the first glance since we do not know a priori the right value of $v(0)$. It means that in fact there are two shooting parameters: $\omega$ and $v(0)$. Even though there exist methods that may work in the case of such two-dimensional shooting \cite{Has}, it is more convenient to get rid of $\omega$ completely by introducing $h(r)=\omega-v(r)$. As a result, one gets the equivalent system of equations
\begin{align}\label{eqn:SNH}
\begin{cases}
    u''+\frac{d-1}{r} u'-r^2 u+hu=0,\\
    h''+\frac{d-1}{r} h'+u^2=0.
\end{cases}
\end{align}
This formal change of variables can be justified as long as $\lim_{r\to\infty}h(r)$ exists. Fortunately, this is the case as can be seen by rewriting the second line of Eq.\ (\ref{eqn:SNH}) into
\begin{align}\label{eqn:h}
h'(r)=-\frac{1}{r^{d-1}}\int_0^r u(s)^2s^{d-1}\, ds.
\end{align}
For stationary solutions, since $u$ vanishes in infinity, this equation leads to $|h(r)|<A r^2$, where $A$ is some constant. As a result, for large $r$, the harmonic term in Eq.\ (\ref{eqn:SNH}) dominates the nonlinear one and $u$ decays exponentially. Then from Eq.\ (\ref{eqn:h}) one sees that $h$ converges to some finite value as needed. In the end we are left with Eq.\ (\ref{eqn:SNH}) together with the initial conditions $u(0)=b$, $h(0)=c$, and $u'(0)=h'(0)=0$. The analysis of this system, will lead us to the following result:

\begin{proposition} \label{P:SNH}
For any $b>0$ there exists a value of $\omega$ such that system (\ref{eqn:SNHt}) has a ground state with $u(0)=b$.
\end{proposition}
\begin{proof}
In this proof we show that for any $b>0$ solutions to Eq. (\ref{eqn:SNH}) with initial conditions $u(0)=b$, $h(0)=c$, and $u'(0)=h'(0)=0$ form a one-parameter family $\{u_c\}$ that satisfies assumptions (A1)--(A6). A similar proof of this Proposition has been presented in \cite{Fic21}. The main goal here is to recast it into the framework introduced by Theorem \ref{T:1}.

We start by investigating the behaviour of the solutions $u_c$ for large values of $c$. Then it is convenient to introduce the rescaled variables $\widetilde{r}=\sqrt{c}r$, $\widetilde{u}_c(\widetilde{r})=u_c(r)$, and $\widetilde{h}_c(\widetilde{r})=h_c(r)/c$.
\begin{align*}
\begin{cases}
    \widetilde{u}_c''+\frac{d-1}{\widetilde{r}} \widetilde{u}_c'-\frac{\widetilde{r}^2}{c^2} \widetilde{u}_c+\widetilde{h}_c\widetilde{u}_c=0,\\
    \widetilde{h}_c''+\frac{d-1}{\widetilde{r}} \widetilde{h}_c'+\frac{1}{c^2}\widetilde{u}_c^2=0.
\end{cases}
\end{align*}
Taking the limit $c\to\infty$ removes two terms from this system and leaves us with equations that can be explicitly solved: $\widetilde{h}_\infty$ is just equal to $1$, while $\widetilde{u}_\infty$ can be expressed with the Bessel function $J_{\frac{d}{2}-1}$ and oscillates indefinitely with decreasing amplitude. It implies that for large enough values of $c$ the solution $u_c$ is crossing zero and monotonically decreasing beforehand, resulting in (A1). 

To prove that (A2) holds, let us assume otherwise: that $u_0$ crosses zero for the first time at some $R>0$. Then multiplication of the first equation in (\ref{eqn:SNH}) by $u_0(r)r^{d-1}$ and integrating over the interval $[0,R]$ leads to some identity. A similar identity can be obtained by multiplying by $u_0'(r)r^d$ and integrating over the same domain. Another two identities can be obtained in an analogous way from the second equation in (\ref{eqn:SNH}) and combining all four of them yields (see \cite{Fic21} for the details)
\begin{align*}
(d-6)\int_0^R u_0'(r)^2 r^{d-1}\, dr+ (d+2)\int_0^R r^2 u_0(r)^2 r^{d-1}\, dr+2u_0'(R)^2 R^d&\\ +h_0'(R)^2R^d+(d-2)h_0(R)h_0'(R)R^{d-1}&=0.
\end{align*}
This Pohozaev-type identity for $d\geq 6$ (i.e., in critical and supercritical dimensions for SNH) consists of purely positive terms on its left-hand side because $h_0(0)=0$ and $h_c(r)$ is decreasing for any $c$ due to Eq.\ (\ref{eqn:h}). We arrive at a contradiction.

Assumption (A3) clearly holds, while (A4) can be checked by a simple analysis of the system (\ref{eqn:SNH}). Additionally, a proper examination of Eq. (\ref{eqn:SNH}) in the limit $r\to 0$ gives $u_c''(0)=-b c/d$ and proves (A5).

Finally, (A6) can be obtained by observing that since $h_c$ is decreasing, for sufficiently large values of $r$ ($r>\sqrt{c}$), the term $-r^2+h_c(r)$ is negative. When it happens, the first line of Eq.\ (\ref{eqn:SNH}) tells us that $u_c$ cannot have positive maxima, nor negative minima. It means that $u_c(r)$ must be monotone from some point on. Then, if $\lim_{r\to\infty} u_c(r)$ exists, it must be equal to zero because otherwise one can calculate the limit of
\begin{align*}
u_c'(r)=\frac{1}{r^{d-1}}\int_0^r \left[s^2-h_c(s)\right] u_c(s)s^{d-1}\, ds
\end{align*}
as $r\to\infty$ using the L'H\^{o}pital's rule and get $\lim_{r\to\infty}|u_c'(r)|=\infty$. It contradicts the convergence of $u_c$ resulting in trichotomy (A6).

Since all necessary assumptions are satisfied, Theorem \ref{T:1} tells us that there exists such a value $c_0$ that $u_{c_0}$ is a positive solution decaying to zero at infinity. It also is the ground state of the initial problem (\ref{eqn:SNHt}) with frequency that can be restored as $\omega=\lim_{r\to\infty}h_{c_0}(r)$.
\end{proof}

\subsection{Gross-Pitaevskii equation} \label{SS:GP}
In the case of Eq.\ (\ref{eqn:gpt}) the stationary solution ansatz and spherical symmetry assumption (justified by \cite{Li93}) lead to the equation
\begin{align}\label{eqn:gp}
    u''+\frac{d-1}{r} u'- r^2 u+u^3+\omega u=0.
\end{align}
Then one has the following result:
\begin{proposition} \label{P:gp}
For any $b>0$ there exists a value of $\omega$ such that the solution $u$ to Eq. (\ref{eqn:gpt}) with $u(0)=b$ is a ground state.
\end{proposition}
Proofs of this Proposition can be found in \cite{Biz21} and \cite{Sel13}. However, in both of these works the authors need to rely on some functional-analytic methods. Theorem \ref{T:1} suggests a more elementary way of obtaining this result.

For Eq.\ (\ref{eqn:gp}) the frequency $\omega$ can be directly used as the shooting parameter $c$, so let $\omega=c$. Most of the assumptions needed for Theorem \ref{T:1} can be checked in a similar way as for SNH. By considering the variables $\widetilde{r}=\sqrt{c}r$, $\widetilde{u}_c(\widetilde{r}) =u_c(r)$ and then taking the limit $c\to\infty$ in Eq.\ (\ref{eqn:gp}), one can prove (A1). Assumption (A2) can again be obtained with the use of the Pohozaev identity, see \cite{Biz21} for the details, but this time it holds for $d\geq 4$ (critical and supercritical dimensions for GP). One can also very simply get (A3), (A4), and (A5). 

Unfortunately, assumption (A6) cannot be proven as simply as before, when one could just use the monotonicity of $h$. Here we can get a better view by introducing new variables $t=r^2/2$ and $w(t)=u(r)/r$ in which Eq.\ (\ref{eqn:gp}) becomes 
\begin{align*}
    \ddot{w}+\frac{d+2}{2t}\dot{w}+w(w^2-1)+\frac{d-1}{4t^2}w+\frac{\omega}{2t}w=0.
\end{align*}
Dots denote here the derivatives in $t$. This system can be interpreted as a description of the damped motion of a point particle in a potential changing its shape from unimodal with a minimum at $w=0$ to W-like with minimas at $w=\pm 1$ and a maximum at $w=0$. This physical picture suggests that the only possible long-time behaviours of the particle are either confinement in one of the two valleys and settling at $w=\pm 1$ or convergence to the maximum at $w=0$. In particular, since the damping term behaves like $t^{-1}$ it should be impossible for $w$ to oscillate indefinitely \cite{Smi61}. However, the strict proof of this fact would require further work. After going back to the original variables, $w\to\pm 1$ would lead to $u\to\pm\infty$, while $w\to 0$ would give $u\to 0$, implying the trichotomy.

Combination of all these conditions would lead, via Theorem \ref{T:1}, to the existence of $c_0$ such that $u_{c_0}$ is the ground state with frequency $\omega=c_0$.

\section{Conclusions} \label{S:conclusions}
An additional question one can ask regarding the obtained solutions is about their uniqueness, i.e., whether for a fixed $b>0$ there is only one value of the shooting parameter $c$ giving the ground state. At this point, no general method of proving this seems to be available. One must instead refer to the case by case analysis. For example, in the case of SNH the uniqueness of the ground state can be proved by methods presented either in \cite{Cho08} or \cite{Gal11} (the second approach was applied in \cite{Fic21}). However, for GP no similar result exists at this point \cite{Biz21} (even though numerical experiments suggest that the obtained ground states are also unique in this case).

The main advantage of the method presented over the other similar approaches \cite{Biz21, Cho08, Sel13}, is that it can be easily expanded to cover also excited states -- stationary solutions that decay to zero at infinity but are crossing zero. Let us just mention here that such states are not bound to be spherically symmetric, so by reduction to ordinary differential equations some solutions are usually lost. Then, to prove the existence of a solution crossing zero exactly once, one can define a set of shooting parameter values in a similar manner as before
\begin{align*}
    I=\{&c\geq 0\, | \, \exists\, 0<r_0<\rho_1<r_1 : u_c(r_0) = u_c(r_1) = 0 \mbox{ and } u_c'(\rho_1) = 0 \mbox{ while, }\\
    & u_c(r) > 0,\, u_c'(r) < 0 \mbox{ for } r \in (0,r_0), u_c(r) < 0,\, u_c'(r) < 0 \mbox{ for } r \in (r_0,\rho_1),\\
    &\mbox{and } u_c(r) < 0,\, u_c'(r) > 0 \mbox{ for } r \in (\rho_1,r_1)  \}.
\end{align*}
This set is non-empty in the cases we covered here since in the limit $c\to\infty$ the solutions were oscillating. This time one needs some better control on the stationary points of the solutions than was needed for the ground state, in particular regarding the emergence of new stationary points from infinity as $c$ changes. Then it is easy to show that $\inf I$ is the sought value of $c$. This idea can be further generalised to any number of crossings with zero by the appropriate choice of the set $I$. Existence of such a ladder of excited states in the case of SNH was shown using this method in \cite{Fic21}.

One can also look for other systems that can be investigated with this approach. Let us point out that the ideas presented above can be easily applied to a broad range of problems with other trapping potentials (not necessarily harmonic) and simple nonlinearities (for example, $|u|^{p-1}u$ where $p>1$). Some early work suggests that similar methods can also work in the case of systems of elliptic equations, such as considered in \cite{Cla22} but in the presence of some trapping potential.

Finally, this research is just a first step in the broader goal of understanding the dynamics of semilinear Schr\"{o}dinger equations with trapping potentials in supercritical dimensions. Some of the results regarding the dependence of frequency $\omega$ on central value $b$ suggest interesting changes in stability of the ground states in higher dimensions \cite{Biz21, Fic21, Pel}. We plan to pursue this direction in the future work.

% BIBLIOGRAPHY

\end{document}